\documentclass[a4paper]{article}
\usepackage{comment}
\usepackage[pdftex]{hyperref}
\usepackage[english]{babel} 

\usepackage{a4wide}

\usepackage[caption=false,font=footnotesize]{subfig}
\usepackage[utf8]{inputenc}

\usepackage[vlined, ruled]{algorithm2e}
\usepackage{amsmath,amsthm,amssymb,makeidx,mathrsfs}
\usepackage{tikz}
\usetikzlibrary{arrows}
\newcommand{\N}{\mathbb{N}}

\newcommand{\A}{\mathcal{A}}

\newcommand{\ie}{\textit{i.e.}}
\newcommand{\etal}{\textit{et al.} }
\newcommand{\Waiting}{\ensuremath{\mathcal{W}}}

\newcommand{\WaitingGreedy}[1]{\ensuremath{\mathcal{WG}_{#1}}}
\newcommand{\cost}[2]{cost_{#1}({#2})}
\newcommand{\whp}{w.h.p. }
\newcommand{\DODA}{\mathcal{D}_\mathsf{ODA}}
\newcommand{\DODAOblivious}{\mathcal{D}_\mathsf{ODA}^{\emptyset}}

\usepackage{bbm}

\newtheorem{lemma}{Lemma}

\newtheorem{theorem}{Theorem}
\newtheorem{corollary}{Corollary}

\begin{document}

\author{Quentin Bramas$^{1}$,
Toshimitsu Masuzawa$^{2}$ and 
S\'ebastien Tixeuil$^{1}$
}

\title{Distributed Online Data Aggregation \\in Dynamic Graphs\thanks{This work was performed within the Labex SMART supported by French state funds managed by the ANR within the Investissements d'Avenir programme under reference ANR-11-IDEX-0004-02.}}
\maketitle

\begin{center}
    $^{1}$ Sorbonne Universit\'es, UPMC Univ Paris 06, France\\
    $^{2}$ Osaka University, Japan
\end{center}

\begin{abstract}
    We consider the problem of aggregating data in a dynamic graph, that is, aggregating the data that originates from all nodes in the graph to a specific node, the sink.  We are interested in giving lower bounds for this problem, under different kinds of adversaries. 
    
    In our model, nodes are endowed with unlimited memory and unlimited computational power. Yet, we assume that communications between nodes are carried out with pairwise interactions, where nodes can exchange control information before deciding whether they transmit their data or not, given that each node is allowed to transmit its data at most once. When a node receives a data from a neighbor, the node may aggregate it with its own data.
    
    We consider three possible adversaries: the online adaptive adversary, the oblivious adversary, and the randomized adversary that chooses the pairwise interactions uniformly at random.
    For the online adaptive and the oblivious adversaries, we give impossibility results when nodes have no knowledge about the graph and are not aware of the future. Also, we give several tight bounds depending on the knowledge (be it topology related or time related) of the nodes.
    For the randomized adversary, we show that the Gathering algorithm, which always commands a node to transmit, is optimal if nodes have no knowledge at all. Also, we propose an algorithm called Waiting Greedy, where a node either waits or transmits depending on some parameter, that is optimal when each node knows its future pairwise interactions with the sink.
\end{abstract}

\section{Introduction}

Dynamic graphs, that is, graphs that evolve over time, can conveniently model dynamic networks, which recently received a lot of interest from the academic community (\emph{e.g.} mobile sensor networks, vehicular networks, disruption tolerant networks, interaction flows, etc.). 
Depending on the problem considered, various models were used: among others, static graphs can be used to represent a snapshot in time of a dynamic graph, functions can be used to define continuously when an edge appears over time, and sequences of tuples can represent atomic interactions between nodes over time.

The problem we consider in this paper assumes an arbitrary dynamic network, such as sensors deployed on a human body, cars evolving in a city that communicate with each other in an ad hoc manner, etc. We suppose that initially, each node in the network originates some data (\emph{e.g.} that originates from a sensor, or from computation), and that these data must be aggregated at some designated node, the \emph{sink}. To this goal, a node may send its data to a communication neighbor at a given time (the duration of this communication is supposed to be one time unit). We assume that there exists an aggregation function that takes two data as input and gives one data as output (the function is aggregating in the sense that the size of the output is supposed to be the same as a single input, such functions include $\min$, $\max$, etc.).
 
The main constraint for communications between nodes is that a node is allowed to send its data (be it its original data, or aggregated data) exactly once (\emph{e.g.} to keep energy consumption low). A direct consequence of this constraint is that a node must aggregate data anytime it receives some, provided it did not send its data previously. It also implies that a node cannot participate to the data aggregation protocol once it has transmitted its data. A nice property of any algorithm implementing this constraint is that the number of communications is minimum. The problem of aggregating all data at the sink with minimum duration is called the \emph{minimum data aggregation time problem}~\cite{bramas2015complexity}. 
The essence of such a data aggregation algorithm is to decide whether or not to send a node's data when encountering a given communication neighbor: by waiting, a node may be able to aggregate more data, while by sending a node disseminates data but excludes itself for the rest of the computation. 

In this paper, we consider that nodes may base their decision on their initial knowledge and past experience (past interactions with other nodes) only. Then, an algorithm accommodating those constraints is called an \emph{online distributed data aggregation} algorithm. The existence of such an algorithm is conditioned by the (dynamic) topology, initial knowledge of the nodes (\emph{e.g.} about their future communication neighbors), etc.

For simplicity, we assume that interactions between the nodes are carried out through pairwise operations. Anytime two nodes $a$ and $b$ are communication neighbors (or, for short, are interacting), either no data transfer happens, or one of them sends its data to the other, that executes the aggregation function on both its previously stored data and the received data, the output is then stored in the (new) stored data of the receiver. 
In the sequel, we use the term \emph{interaction} to refer to a pairwise interaction.

We assume that an adversary controls the dynamics of the network, that is, the adversary decides which are the interactions. As we consider atomic interactions, the adversary decides what sequence of interactions is to occur in a given execution. Then, the sequence of static graphs to form the evolving graph can be seen as a sequence of single edge graphs, where the edge denotes the interaction that is chosen by the scheduler at this particular moment. Hence, the time when an interaction occurs is exactly its index in the sequence. Our model of dynamic graphs as a sequence of interactions differs from existing models on several points. First, general models like  \textit{Time-varying-graph}~\cite{casteigts2011time} make use of continuous time, which adds a lot of complexity. Also, discrete time general models such as \emph{evolving graph}~\cite{casteigts2011time} capture the network evolution as a sequence of static graphs. Our model is a simplification of the evolving graph model where each static graph has a single edge. \textit{Population protocols}~\cite{angluin2007thecumputational} also consider pairwise interactions, but focus on finite state anonymous nodes with limited computational power and unlimited communication power (a given node can transmit its information many times), while we consider powerful nodes (that can record their past interactions) that are communication limited (they can send their data only once). Finally, \textit{Dynamic edge-relabeling}~\cite{casteigts2010srtuctural} is similar to population protocols, but the sequence of pairwise interactions occurs inside an evolving graph. This model shares the same differences as population protocols with our model.

 \subsection{Related Work}
The problem of data aggregation has been widely studied in the context of wireless sensor networks. The literature on this problem can be divided in two groups depending on the assumption made about the collisions being handled by an underlying MAC layer.
 
\textit{In the case when collisions are not handled by the MAC layer}, the goal is to find a collision-free schedule that aggregates the data in minimum duration. The problem was first studied by Annamalai \etal~\cite{annamalai2003tree}, and formally defined by Chen \etal~\cite{chen2005minimum}, which proved that the problem is NP-complete. Then, several papers~\cite{yu2009distributed, xu2011delay, ren2010new, nguyen2011efficient} proposed centralized and distributed approximation algorithms for this problem. The best known algorithm is due to Nguyen \etal~\cite{nguyen2011efficient}. More recently, Bramas \etal~\cite{bramas2015complexity} considered the generalization of the problem to dynamic wireless sensor networks (modeled by evolving graphs). Bramas \etal~\cite{bramas2015complexity} show that the problem remains NP-complete even when restricted to dynamic WSNs of degree at most $2$ (compared to $3$ in the static case).
  
\textit{When collisions are handled by the MAC layer}, various problems related to data aggregation have been investigated. The general term \emph{in-network aggregation} includes several problems such as gathering and routing information in WSNs, mostly in a practical way. For instance, a survey~\cite{fasolo2007network} relates aggregation functions, routing protocols, and MAC layers with the objective of reducing resource consumption. \emph{Continuous aggregation}~\cite{abshoff2014continuous} assumes that data have to be aggregated, and that the result of the aggregation is then disseminated to all participating nodes. The main metric is then the delay before aggregated data is delivered to all nodes, as no particular node plays the role of a sink. Most related to our concern is the work by Cornejo \etal~\cite{cornejo2012aggregation}. In their work, each node starts with a token, the time is finite and no particular node plays the role of a sink node. Then, the topology evolves with time, and at each time instant, a node has at most one neighbor with which it can interact and send or not its token. The goal is to minimize the number of nodes that own at least one token. Assuming an algorithm does not know the future, Cornejo \etal~\cite{cornejo2012aggregation} prove that its competitive ratio is $\Omega(n)$ with high probability (w.r.t. the optimal offline algorithm) against an oblivious adversary. 

\subsection{Our Contributions}
In this paper we define the problem of distributed online data aggregation in dynamic graphs, and study its complexity. It turns out that the problem difficulty strongly depends on the power of the adversary (that chooses which interactions occur in a given execution).

For the oblivious and the online adaptive adversaries, we give several impossibility results when nodes have no knowledge about the future evolution of the dynamic graph, nor about the topology. Also, when nodes are aware of the underlying graph (where an edge between two nodes exists if those nodes interact at least once in the execution), the data aggregation is impossible in general.
To examine the possibility cases, we define a cost function whose purpose is to compare the performance of a distributed online algorithm to the optimal offline algorithm for the same sequence of interactions. Our results show that if all interactions in the sequence occur infinitely often, there exists a distributed online data aggregation algorithm whose cost is finite. Moreover, if the underlying graph is a tree, we present an optimal algorithm.

For the randomized adversary, we first present tight bounds when nodes have full knowledge about the future interactions in the whole graph. In this case, the best possible algorithm terminates in $\Theta(n\log(n))$ interactions, in expectation and with high probability. 
Then, we consider nodes with restricted knowledge, and we present two optimal distributed online data aggregation algorithms that differ in the knowledge that is available to nodes. The first algorithm, called \emph{Gathering}, assumes nodes have no knowledge whatsoever, and terminates in $O(n^2)$ interactions on average, which we prove is optimal with no knowledge. The second one, called \emph{Waiting Greedy}, terminates in $O\left(n^{3/2}\sqrt{\log(n)}\right)$ interactions with high probability, which we show is optimal when each node only knows the time of its next interaction with the sink (the knowledge assumed by Waiting Greedy).

We believe our research paves the way for stimulating future researches, as our proof arguments present techniques and analysis that can be of independent interest for studying dynamic networks.


\section{Model}

A dynamic graph is modeled as a couple $(V, I)$, where $V$ is a set of nodes and $I = \left(I_t\right)_{t\in \N}$ is a sequence of pairwise interactions (or simply interactions). A special node in $V$ is the \emph{sink} node, and is denoted by $s$ in the sequel. In the sequence $\left(I_t\right)_{t\in \N}$, the index $t$ of an interaction also refers to its \emph{time of occurrence}. In the sequel $V$ always denotes the set of nodes, $n\geq 3$ its size, and $s\in V$ the sink node. 

In general, we consider that nodes in $V$ have unique identifiers, unlimited memory and unlimited computational power. However, we sometimes consider nodes with no persistent memory between interactions; those nodes are called \emph{oblivious}. 

Initially, each node in $V$ receives a data. During an interaction $I_t = \{u,v\}$, if both nodes still own data, then one of the node has the possibility to transmit its data to the other node. The receiver aggregates the received data with its own data. The transmission and the aggregation take exactly one time unit. If a node decides to transmit its data, then it does not own any data, and is not able to receive other's data anymore.

\subsection{Problem Statement}
The data aggregation problem consists in choosing at each interaction whether a node transmits (and which one) or not so that after a finite number of interactions, the sink is the only node that owns a data.
In this paper we study distributed and online algorithms that solve this problem. Such algorithms are called
\emph{distributed online data aggregation} (DODA) algorithms. 

A DODA is an algorithm that takes as input an interaction $I_t = \{u,v\}$, and its time of occurrence $t\in \N$, and outputs either $u$, $v$ or $\bot$. If a DODA outputs a node, this node is the receiver of the other node's data. In more details, if $u$ is the output, this means that before the interaction both $u$ and $v$ own data, and the algorithm orders $v$ to transmit its data to $u$. The algorithm is able to change the memory of the interacting nodes, for instance to store information that can be used in future interactions. In the sequel, $\DODA$ denotes the set of all DODA algorithms. And $\DODA^{\emptyset}$ denotes the set of DODA algorithms that only require oblivious nodes.

A DODA can require some knowledge to work. A knowledge is a function (or just an attribute) given to every node that gives some information about the future, the topology or anything else. By default, a node $u\in V$ has two information: its identifier $u.ID$ and a boolean $u.isSink$ that is true if $u$ is the sink, and false otherwise. A DODA algorithm may use additional functions associated with different knowledge. $\DODA(\mathfrak{i}_1, \mathfrak{i}_2, \ldots)$ denotes the set of DODA algorithms that use the functions $\mathfrak{i}_1, \mathfrak{i}_2, \ldots$. For instance, we define for a node $u\in V$ the function $u.meetTime$ that maps a time $t\in\N$ with the smallest time $t'>t$ such that $I_{t'} = \{u,s\}$ \ie, the time of the next interaction with the sink (for $u=s$, we define $s.meetTime$ as the identity, $t\mapsto t$). Then $\DODA(meetTime)$ refers to the set of DODA algorithms that use the information $meetTime$.

\subsection{Adversary Models}

In this paper we consider three models of adversaries:
\begin{itemize}
    \item The oblivious adversary. This adversary knows the algorithm's code, and must construct the sequence of interactions before the execution starts. 

    \item This adversary knows the algorithm's code and can use the past execution of the algorithm to construct the next interaction. However, it must make its own decision as it does not know in advance the decision of the algorithm. In the case of deterministic algorithms, this adversary is equivalent to the oblivious adversary.
    
    \item The randomized adversary. This adversary constructs the sequence of interactions by picking pairwise interactions uniformly at random.
\end{itemize}
Section \ref{sec:oblivious adversary} presents our results with the oblivious and the adaptive online adversary. The results with the randomized adversary are given in section \ref{sec:random adversary}.

\subsection{Definition of Cost}
To study and compare different DODA algorithms, we use a tool slightly different from the competitive analysis that is generally used to study online algorithms. The competitive ratio of an algorithm is the ratio between its performance and the optimal offline algorithm's performance. However, one can hardly define objectively the performance of an algorithm. For instance, if we just consider the number of interactions before termination, then an oblivious adversary can construct a sequence of interactions starting with the same interaction repeated an arbitrary number of time. In this case, even the optimal algorithm has infinite duration. Moreover, the adversary can choose the same interaction repeatedly after that the optimal offline algorithm terminates. This can prevent any non optimal algorithm from terminating and make it have an infinite competitive-ratio. 

To prevent this we define the cost of an algorithm. Our cost is a way to define the performance of an algorithm, depending on the performance of the optimal offline algorithm. We believe our definition of cost is well-suited for a lots of problems where the adversary has a strong power, especially in dynamic networks. One of its main advantage is that it is invariant by trivial transformation of the sequence of interactions, like inserting or deleting duplicate interactions.

For the sake of simplicity, a data aggregation schedule with minimum duration (performed by an offline optimal algorithm) is called a \emph{convergecast}.
Consider a sequence of interactions $I$. Let $opt(t)$ be the ending time of a convergecast on $I$, starting at time $t\in\N$. If the ending time is infinite (if the optimal offline algorithm does not terminate) we write $opt(t)=\infty$. Let $T:\N_{\ge 1} \mapsto \N\cup\{\infty\}$ be the function defined as follow:

\begin{align*}
    T(1) &= opt(0)\\
    \forall i\geq1\quad T(i+1) &= opt(T(i) + 1)\\ 
\end{align*}
$T(i)$ is the duration of $i$ successive convergecasts (two convergecasts are consecutive if the second one starts just after the first one completes).

Let $duration(A,I)$ be the termination time of algorithm $A$ executed on the sequence of interactions $I$.
Now, we define the cost $\cost{A}{I}$ of an algorithm $A$ on the sequence $I$, as the smallest integer $i$ such that $duration(A,I)\leq T(i)$:
\[
    \cost{A}{I} = \min \{i\;|\; duration(A,I)\leq T(i)\} 
\]
This means that $\cost{A}{I}$ is an upper bound on the number of successive convergecasts we can perform during the execution of $A$, on the sequence $I$.
It follows from the definition that an algorithm performs an optimal data aggregation if and only if $\cost{A}{I} = 1$. 

Also, if $duration(A,I)=\infty$, then it is possible that $\cost{A}{I}<\infty$. Indeed, if $i_{\max} = \min_i\{i\,|\,T(i) = \infty\}$ is well-defined, then $\cost{A}{I} = i_{\max}$, otherwise $\cost{A}{I}=\infty$.

%
%

\section{Oblivious and Online Adaptive Adversaries}\label{sec:oblivious adversary}
In this section we give several impossibility results when nodes have no knowledge, and then show several results depending on the amount of knowledge. We choose to limit our study to some specific knowledge, but one can be interested in studying the possible solutions for different kind of knowledge.

\subsection{Impossibility Results When Nodes Have no Knowledge}

\begin{theorem}
    For every algorithm $A\in \DODA$, there exists an adaptive online adversary generating a sequence of interactions $I$ such that $\cost{A}{I} = \infty$.
\end{theorem}
\begin{proof}
    Let $I$ the sequence of interactions between $3$ nodes $a$, $b$, and the sink $s$, defined as follows.
    $I_0 = \{a,b\}$. If $a$ transmits, then for every $i\in\N$, $I_{2i+1} = \{a,s\}$ and $I_{2i+2} = \{a,b\}$ so that $b$ will never be able to transmit.
    Symmetrically if $b$ transmits the same thing happens.
    If no node transmits, then $I_1 = \{b,s\}$. If $b$ transmits, then $I_{2i+2} = \{a,b\}$ and $I_{2i+3} = \{b,s\}$ so that $a$ will never be able to transmit.
    Otherwise $I_2 = \{a,b\}$ and continue as in the first time.
    $A$ never terminates, and a convergecast is always possible, so that $\cost{A}{I} = \infty$.
\end{proof}

In the case of deterministic algorithm, the previous theorem is true even with an oblivious adversary. However, for a randomized algorithm, the problem is more complex. The following theorem states that the impossibility results for oblivious randomized algorithm, leaving the case of general randomized algorithms against oblivious adversary as an open question.


\begin{theorem}
    For every randomized algorithm $A\in \DODA^\emptyset$, there exists an oblivious adversary generating a sequence of interactions $I$ such that $\cost{A}{I} = \infty$ with high probability\footnote{An event $A$ occurs with high probability, when $n$ tends to infinity, if $P(A) > 1-o\left(\dfrac{1}{\log(n)}\right)$}.
\end{theorem}
\begin{proof}
    Let $V=\{s, u_0, \ldots, u_{n-2}\}$. In the sequel, indexes are modulo $n-1$ \ie,  $\forall i,j\geq0$, $u_i=u_j$ with $i\equiv j \mod n-1$. Let $I^\infty$ defined by, for all $i\in\N$, $I^\infty_i = \{u_i, s\}$.
    Let $I^l$ be the finite sequence, prefix of length $l>0$ of $I^\infty$.
    For every $l>0$, the adversary can compute the probability $P_l$ that no node transmits its data when executing $A$ on $I^l$. $(P_l)_{l>0}$ is a non-increasing sequence, it converges to a limit $\mathcal{P}\geq 0$.
    For a given $l$, if $P_l\geq 1/n$, there is at least two nodes whose probability not to transmit when executing $A$ on $I^l$ is at least $n^{-\frac{1}{n-2}} = 1 - O\left(\frac{1}{\sqrt{n}}\right)$.
To prove this, we can see the probability $P_l$ as the product of $n-1$ probabilities $p_0$,  $p_1$, $\ldots$, $p_{n-2}$ where $p_i$ is the probability that node $u_i$ does not transmit during $I^l$. Those events are independent since the algorithm is oblivious. Let $p_{d}\geq p_{d'}$ be the two greatest probabilities in $\{p_i\}_{0\leq i\leq n-2}$, we have:
    \[\left(\prod_{i=0}^{n-2} p_i \geq \frac{1}{n}\right)
                     \Rightarrow \left(\sum_{i=0}^{n-2} \log(p_i) \geq \log\left(\frac{1}{n}\right)\right)
        \Rightarrow  \left((n-2)\log(p_{d'}) \geq\log\left(\frac{1}{n}\right)\right)
        \Rightarrow \left(p_{d'} \geq n^{-\frac{1}{n-2}}\right)
\]    
    
    This implies that, if $\mathcal{P}\geq 1/n$, then $A$ does not terminate on the sequence $I^{\infty}$ with high probability.
    
    Otherwise, let $l_0$ be the smallest index such that $P_{l_0} < 1/n$. So that with high probability, at least one node transmits when executing $A$ on $I^{l_0}$. 
    Also, $P_{l_0 - 1}\geq 1/n$ so that the previous argument implies that there is at least two nodes $u_d$ and $u_{d'}$ whose probability to still have a data (after executing $A$ on $I^{l_0 - 1}$) is at least $n^{-\frac{1}{n-2}}$. If $l_0 = 0$ we can choose $\{u_d, u_{d'}\} = \{u_1, u_2\}$. We have $u_d\neq u_{l_0}$ or $u_{d'}\neq u_{l_0}$. Without loss of generality, we can suppose $u_d\neq u_{l_0}$, so that the probability that $u_d$ transmits is the same in $I^{l_0 - 1}$ and in $I^{l_0}$.
    
    Now, $u_d$ is a node whose probability not to transmit when executing $A$ on $I^{l_0}$ is at least $n^{-\frac{1}{n-2}} = 1 - O\left(\frac{1}{\sqrt{n}}\right)$. Let $I'$ be the sequence of interactions defined as follow:
    \[
        \forall i\in [0,n-2]\setminus\{d-1\},\; I'_i = \{u_{i}, u_{i+1}\},\;
        I'_{d-1} = \{u_{d-1}, s\}
    \]
     
    $I'$ is constructed such that $u_d$ (the node that has data with high probability) must send its data along a path that contains all the other nodes in order to reach the sink. But this path contains a node that does not have a data.

    Let $I$ be the sequence of interaction starting with $I^{l_0}$ and followed by $I'$ infinitely often.
    We have shown that with high probability, after $l_0$ interactions, at least one node transmits its data and the node $u_d$ still has a data. The node that does not have data prevents the data owned by $u_d$ from reaching $s$. So that $A$ does not terminate, and since a convergecast is always possible, then $\cost{A}{I} = \infty$.
\end{proof}

\subsection{When Nodes Know The Underlying Graph}

Let $\bar{G}$ be the underlying graph \ie, $\bar{G}=(V,E)$ with $E=\left\{(u,v)\,|\,\exists t \in \N, \, I_t = \{u,v\}\right\}$. The following results assume that the underlying graph is given initially to every node.
\begin{theorem}\label{thm: impossibility results with knowledge of the underlying graph agains online adversary}
    If $n\geq 4$, then, for every algorithm $A\in \DODA(\bar{G})$, there exists an online adaptive adversary generating a sequence of interactions $I$ such that $\cost{A}{I} = \infty$.
\end{theorem}
\begin{proof}
    $V = \{s, u_1, u_2, u_3\}$. We create a sequence of interactions with the underlying graph $\bar{G}=\left(V,\left\{(s, u_1), (u_1, u_2),(u_2, u_3), (u_3, s)\right\}\right)$. We start with the following interactions: 
    \begin{equation}\label{eq:interaction to fail DODA with underlying graph}
    \left(\{u_1, s\},\{u_3, s\},\{u_2, u_1\},\{u_2, u_3\}\right).
    \end{equation}    
    If $u_2$ transmits to $u_1$ in $I_2$, then we repeat infinitely often the three following interactions:
    \[\left(\{u_1, u_2\},\{u_2, u_3\},\{u_3, s\}, ...\right).\]
    Else, if $u_2$ transmits to $u_3$ in $I_3$, then we repeat infinitely often the three following interactions:
    \[\left(\{u_3, u_2\},\{u_2, u_1\},\{u_1, s\}, ...\right).\]
    Otherwise, we repeat the four interactions (\ref{eq:interaction to fail DODA with underlying graph}), and apply the previous reasoning. Then, $A$ never terminates, and a convergecast is always possible, so that $\cost{A}{I} = \infty$.
\end{proof}

\begin{theorem}
    If the interactions occurring at least once, occur infinity often, then there exists $A \in \DODAOblivious(\bar{G})$ such that $\cost{A}{I} < \infty$ for every sequence of interactions $I$. However, $\cost{A}{I}$ is unbounded.
\end{theorem}
\begin{proof}
    Nodes can compute a spanning tree $T$ rooted at $s$ (they compute the same tree, using nodes identifiers). Then, each node waits to receive the data from its children and then transmits to its parent as soon as possible. All transmissions are done in finite time because each edge of the spanning tree appears infinitely often.
    However, when $\bar{G}$ is not a tree, there exists another spanning tree $T'$. Let $e$ be an edge of $T$ that is not in $T'$. By repeated interactions of edges of $T'$, an arbitrary amount of convergecasts can be performed while a node is waiting for sending data to its parent through $e$ in execution of $A$.
\end{proof}
\begin{theorem}
    If $\bar{G}$ is a tree, there exists $A\in \DODAOblivious(\bar{G})$ that is optimal.
\end{theorem}
\begin{proof}
    Each node waits to receive the data from its children, then transmits to its parent as soon as possible.
\end{proof}

\subsection{If Nodes Know Their Own Future}
For a node $u\in V$, $u.future$ denotes the future of $u$ \ie, the sequence of interactions involving $u$, with their times of occurrences. In this case, according to the model, two interacting nodes exchange their future and non-oblivious nodes can store it. This may seem in contradiction with the motivation of the problem that aims to reduce the number of transmissions. However, it is possible that the data must be sent only once for reasons not related to energy (such as data that cannot be duplicated, tokens, etc.). That is why we consider this case, for the sake of completeness, even if oblivious algorithms should be favored.
\begin{theorem}
    There exists $A\in \DODA(future)$ such that $\cost{A}{I} \leq n$ for every sequence of interactions $I$.
\end{theorem}
\begin{proof}
One can show that the duration of $n-1$ successive convergecasts is sufficient to perform a broadcast from any source. So every node broadcasts its future to the other nodes. After that, all the nodes are aware of the future of every node and can compute the optimal data aggregation schedule. So that it takes only one convergecast to aggregate the data of the whole network.  In total, $n$ successive convergecasts are sufficient.
\end{proof}

%
%

\section{Randomized Adversary}\label{sec:random adversary}

The randomized adversary constructs the sequence of interactions by picking a couple of nodes among all possible couples, uniformly at random. Thus, the underlying graph is a complete graph
of $n$ nodes (including the sink) and every interaction occurs with 
the same probability $p = \frac{2}{n(n-1)}$.

In this section, the complexity is computed on average (because the adversary is randomized) and no more ``in the worst case'' as previously. In this case, considering the number of interactions is sufficient to represent the complexity of an algorithm. We see in Theorem~\ref{thm:performance of offline algorithm agains randomized adversary} that an offline algorithm terminates in $\Theta(n\log(n))$ interactions w.h.p. This bound gives a way to convert the complexity in terms of number of interactions to a cost. Indeed, if an algorithm $\A$ terminates in $O(n^2)$ interactions, then its performance is $O(n/\log(n))$ times worse than the offline algorithm and $\cost{A}{I} = O(n/\log(n))$ for a randomly generated sequence of interactions $I$. For the sake of simplicity, in the remaining of the section, we give the complexity in terms of number of interactions.

Since an interaction does not depend on previous interactions, the algorithms we propose here are oblivious i.e., they do not modify the memory of the nodes. In more details, the output of our algorithms depends only on the current interaction and on the information available in the node.

First, we introduce three oblivious DODA algorithms. For the sake of simplicity, we assume that the output is ignored if the interacting nodes do not both have data. Also, to break symmetry, we suppose the nodes that interact are given as input ordered by their identifiers.

\begin{itemize}    
\item Waiting ($\Waiting\in \DODAOblivious$): A node transmits only when it is connected to the sink $s$: 
    \[
        \Waiting: (u_1, u_2, t)=\left\{
        \begin{array}{ll}
    u_i\,&\text{if $u_i.isSink$}\\
    \bot\,&\text{otherwise}
\end{array}        
        \right.
    \]
\item
Gathering (${\cal GA} \in \DODAOblivious$): A node transmits its data
when it is connected to the sink $s$ or to a node having data.:
\[
  {\cal GA}: (u_1, u_2, t) = \left\{ \begin{array}{ll}
  u_i & {\rm if}\ u_i.is.Sink \\
  u_1 & {\rm otherwise} \\
 \end{array} \right.
\]
    \item Waiting Greedy with parameter $\tau\in\N$ ($\WaitingGreedy{\tau}\in \DODAOblivious(meetTime)$): The node with the greatest meet time transmits, if its meet time is greater than $\tau$:

    \begin{align*}
        m_1 = u_1.meetTime(t) \\
        m_2 = u_2.meetTime(t)
    \end{align*}
        \[\arraycolsep=2.5pt\def\arraystretch{1}
                \WaitingGreedy{\tau}: (u_1, u_2, t){=}\left\{
                \begin{array}{ll}
            {u_1}&\text{if $m_1 \leq m_2\wedge \tau < m_2$}\\
            {u_2}&\text{if $m_1 > m_2\wedge \tau < m_1$}\\
            {\bot}&\text{otherwise}
        \end{array}        
                \right.
        \]
One can observe that after time $\tau$, the algorithm acts as the Gathering algorithm.
\end{itemize}
\subsection{Lower Bounds}

We show a lower bound $\Omega (n^2)$ on the number of interactions 
required for DODA against the randomized adversary. The lower bound holds for all algorithms (including randomized
ones) that do not have knowledge about future of the evolving network.
The lower bound matches the upper bound of the \emph{Gathering} algorithm given in the next subsection. This implies that this bound is tight.

\begin{theorem}
The expected number of interactions required for DODA is $\Omega(n^2)$.
\end{theorem}

\begin{proof}
We show that the last data transmission requires $\Omega(n^2)$ 
interactions in expectation.

We consider any (randomized) algorithm {\cal A} and its execution
for DODA.  Before the last transmission (from some node, say $v$, 
to the sink $s$), only $v$ has data except for $s$.

The probability that $v$ and $s$ interacts in the next interaction is 
$\frac{2}{n(n-1)}$.  
Thus, the expected number $EI$ of interactions required for $v$ to transmit to $s$ is:
\[ EI = \frac{n(n-1)}{2}  \]
So that the whole aggregation requires at least $EI=\Omega(n^2)$.
\end{proof}

We also give a tight bound for algorithms that know the full sequence of interactions.

\begin{theorem}\label{thm:performance of offline algorithm agains randomized adversary}
The best algorithm in $\DODAOblivious$(full knowledge) terminates in $\Theta(n \log(n))$ interactions, in expectation and with high probability.
\end{theorem}

\begin{proof}
First, we show that the expected number of interactions of a broadcast algorithm is $\Theta(n \log(n))$.
The first data transmission occurs when the source node (say $v_0$) 
interacts with another node.  
The probability of occurrence of the first data transmission is 
$\frac{2(n-1)}{n(n-1)}$.  
After the $(i-1)$-th data transmission, $i$ nodes (say 
$V_{i-1}=\{v_0, v_1, \ldots , v_{i-1}\}$) have the data and 
the $i$-th data transmission occurs when a node in 
$V_{i-1}$ interacts with a node not in $V_{i-1}$. This happens with probability $\frac{2i(n-i)}{n(n-1)}$.

Thus, if $X$ is the number of interactions required to perform a broadcast,  then we have:
  \begin{align*}
    E(X)&=\sum^{n-1}_{i=1} \frac{n(n-1)}{2i(n-i)} 
     =\frac{n(n-1)}{2} \sum^{n-1}_{i=1} \frac{1}{i(n-i)} \\
     &=\frac{n(n-1)}{2n} \sum^{n-1}_{i=1} (\frac{1}{i}+\frac{1}{n-i}) \\
     &= (n-1) \sum^{n-1}_{i=1} \frac{1}{i} \in \Theta(n \log(n)).
  \end{align*}
  
And the variance is 
\begin{align*}
    Var(X) &= \sum^{n-1}_{i=1} \left(1-\frac{2i(n-i)}{n(n-1)} \right)/\left(\frac{2i(n-i)}{n(n-1)}\right)^2 \\
    &= n(n-1)\sum^{n-1}_{i=1} \frac{n(n-1) - 2i(n-i)}{\left(2i(n-i)\right)^2}\\
    &=O\left(n^4\sum^{\lfloor n/2\rfloor-1}_{i=1} \left(\frac{1}{i(n-i)}\right)^2\right)
\end{align*}
The last sum is obtained from the previous one by observing that it is symmetric with respect to the index $i=n/2$, and the removed elements ($i=\lfloor n/2 \rfloor$ and possibly $i=\lceil n/2 \rceil$) are negligible.
We define $f: x\mapsto \frac{1}{x^2(n-x)^2}$. Since $f$ is increasing between $1$ and $n/2$, we have
\begin{align*}
    \sum^{\lfloor n/2\rfloor - 1}_{i=1} f(i) &\leq \int_{1}^{n/2}f(x)dx \\
    &= \frac{\frac{(n-2) n}{n-1}+2 \log (n-1)}{n^3} {=} O\left(\frac{1}{n^2}\right)
\end{align*}
So that the variance is in $O(n^2)$.
Using the Chebyshev's inequality, we have
\begin{align*}
    P(|X - E(X)| > n\log(n))
    &=O\left(\frac{1}{\log^2(n)}\right)
\end{align*}
Therefore, a sequence of $\Theta(n\log(n))$ interactions is sufficient to perform a broadcast with high probability.
By reversing the order of the interactions in the sequence of interactions, this implies that a sequence of $\Theta(n\log(n))$ interactions is also sufficient to perform a convergecast with the same probability. Aggregating data along the convergecast tree gives a valid data aggregation schedule.
\end{proof}

\begin{corollary}
The best algorithm in $\DODA(future)$ terminates in $\Theta(n \log(n))$ interactions, in expectation and with high probability.
\end{corollary}
\begin{proof}
If each node starts with its own future, $O(n\log(n))$ interactions are sufficient to retrieve with high probability the future of the whole network. Then $O(n\log(n))$ interactions are sufficient to aggregate all the data with the full knowledge.
\end{proof}

\subsection{Algorithm Performance Without Knowledge}

\tolerance=2000
\begin{theorem}
The expected number of interactions the Waiting
requires to terminate is $O(n^2 \log(n))$.

The expected number of interactions the Gathering
requires to terminate is $O(n^2)$.
\end{theorem}

\begin{proof}
In the {\emph Waiting} algorithm, data is sent to the sink when a node with data is connected to the sink. We denote by $X_W$ the random variable that equals the number of interactions for the algorithm Waiting to terminate.
The probability of occurrence of the first data transmission is 
$\frac{2(n-1)}{n(n-1)}$.  
The probability of occurrence of the $i$-th data transmission 
after the $(i-1)$-th data transmission is $\frac{2(n-i)}{n(n-1)}$.  
Thus, the expected number of interactions required for DODA is
  \begin{align*}
    E(X_W) &= \sum^{n-1}_{i=1} \frac{n(n-1)}{2(n-i)} \\
    &= \frac{n(n-1)}{2} \sum^{n-1}_{i=1} \frac{1}{i} \in O(n^2 \log(n))
  \end{align*}
Since those events are independent, we also have that the variance of the number of interactions required for DODA is
\begin{align*}
    Var(X_W)
    &=\sum^{n-1}_{i=1} \frac{n(n-1) - 2i}{n(n-1)}\times\frac{(n(n-1))^2}{4i^2}\\
    &=\sum^{n-1}_{i=1} \frac{n^2(n-1)^2 - 2in(n-1)}{4i^2}\\
    &\sim_{+\infty} \sum^{n-1}_{i=1} \frac{n^4}{4i^2}\qquad\sim_{+\infty} \frac{n^4\pi^2}{24}
\end{align*}
Using the Chebyshev's inequality, we have
\begin{align*}
    P(\left|X_W - E(X_W)\right| > n^2\log(n)) 
    &= O\left( \frac{n^4\pi^2}{24n^4\log^2(n)}\right)\\
    &= O\left( \frac{1}{\log^2(n)} \right)
\end{align*}
Therefore, algorithm Waiting terminates after $O(n^2\log(n))$  interactions with probability greater than $1-1/log^2(n)$.

In the Gathering algorithm, a data is sent when a node with the data is connected to the sink or another node with data. We denote by $X_G$ the random variable that equals the number of interactions for the algorithm Gathering to terminate.
Notice that the total number of data transmissions required to terminate
is exactly $n-1$.
The probability of occurrence of the first data transmission is 
$\frac{n(n-1)}{n(n-1)}=1$.  
The probability of occurrence of the $i$-th data transmission 
after the $(i-1)$-th data transmission is $\frac{(n-i+1)(n-i)}{n(n-1)}$.  
Thus, the expected number of interactions required to terminate is
  \begin{align*}
    E(X_G) &= \sum^{n-1}_{i=1} \frac{n(n-1)}{(n-i+1)(n-i)} \\
    &= n(n-1) \sum^{n-1}_{i=1} \frac{1}{i(i+1)} \in O(n^2)
  \end{align*}
\end{proof}
\begin{corollary}\label{thm:optimality of gathering}
    Algorithm Gathering is optimal in $\DODA$.
\end{corollary}

\subsection[Algorithm Performance With meetTime]{Algorithm Performance With $\mathit{meetTime}$}
In this subsection we study the performance of our algorithm Waiting Greedy, find the optimal value of the parameter $\tau$ and prove that this is the best possible algorithm with only the $meetTime$ information (even if nodes have unbounded memory).
We begin by a lemma to find how many interactions are needed to have a given number of nodes interacting with the sink.
\begin{lemma}\label{lemme: nf(n) interactions allow f(n) nodes to interact with the sink}
    If $f$ is a function such that $f(n) = \omega(1)$ and $f(n) = o(n)$ then, in $nf(n)$ interactions, $\Theta(f(n))$ nodes interact with the sink with high probability.
\end{lemma}

\begin{proof}
    The probability of the $i$-th interaction between the sink and a node that has a data, after $i-1$ such interactions, is $\frac{2(n-i)}{n(n-1)}$. Let $X$ be the number of interactions needed for the sink to meet $f(n)$ different nodes. We have:
    
     \begin{align*}
        E(X) &=\sum^{f(n)}_{i=1} \frac{n(n-1)}{2(n-i)} \\
        &=\frac{n(n-1)}{2} (H(n - 1) - H(n-f(n))) \\
        &\sim \frac{n^2}{2} \left(-\log\left(1-\frac{f(n)}{n}\right)+ o(1)\right)\\
        &\sim \frac{n^2}{2} \frac{f(n)}{n}\sim \frac{f(n)n}{2}
    \end{align*}
    
    and the variance is
    \begin{align*}
        Var(X)&=\sum^{f(n)}_{i=1} \left(1-\frac{2(n-i)}{n(n-1)}\right)/\left(\frac{2(n-i)}{n(n-1)}\right)^2\\
    &\sim \sum^{f(n)}_{i=1} \frac{n^4}{4n^2} \sim \frac{n^2}{4}f(n)
    \end{align*}   
    Using the Chebyshev's inequality, we have
    \begin{align*}
        P(|X - E(X)| > nf(n)) &=O\left( \frac{n^2f(n)}{4n^2f(n)^2}\right)\\
        &=O\left( \frac{1}{f(n)}\right)
    \end{align*}
    So that $X = \Theta\left(nf(n)\right)$ with high probability if $1/f(n) = o(1)$ (or equivalently $f(n) = \omega(1)$).
\end{proof}
Now we can state our theorem about the performance of Waiting Greedy depending on the parameter $\tau$.
\begin{theorem}
    Let $f$ be a function such that  $f(n) = o(n)$ and $f(n) = \omega(1)$. The algorithm Waiting Greedy with $\tau = \Theta\left(\max\left(nf(n), n^2\log(n)/f(n)\right)\right)$ terminates in $\tau$ interactions with high probability.
\end{theorem}

\begin{proof}
To have an upper bound on the number of interactions needed by Waiting Greedy to terminate, we decompose the execution in two phases, one between time 0 and a time $t_1$ and the other between time $t_1$ and a time $t_2=\tau$. In the last phase, a set of nodes $L\subset V$ interacts at least once directly with the sink. Nodes in $L$ do not transmit to anyone in the first phase by definition of the algorithm (they have a meetTime smaller than $\tau$). Nodes in $L$ help the other nodes (in $L^c = V\backslash L$) to transmit their data in the first phase. Maybe nodes in $L^c$ can transmit to $L$ in the second phase, but we do not take this into account, that is why it is an upper bound. 

If a node $u$ in $L^c$ interacts with a node in $L$ in the first phase, either it transmits its data, otherwise (by definition of the algorithm) it has a meetTime smaller than $\tau$ (and smaller than $t_1$ because it is not in $L$). In every case, a node in $L^c$ that meets a node in $L$ in the first phase, transmits its data.
To prove the theorem i.e., in order for the algorithm to terminate before $\tau$ with high probability, we prove two claims: (a) the number of nodes in $L$ is $f(n)$ with high probability if $t_2 - t_1 = nf(n)$ and (b) all nodes in $L^c$ interact with a node in $L$ with high probability if $t_1 = \Theta(n^{2}\log(n)/f(n))$.
The first claim is implied by Lemma~\ref{lemme: nf(n) interactions allow f(n) nodes to interact with the sink}. Now we prove the second claim.

    Let $X$ be the number of interactions required for the nodes in $L^c$ to meet a node in $L$.
    The probability of the $i$-th interaction between a node in $L^c$ (with a data) and a node in $L$, after $i-1$ such interactions already occurred, is ${2f(n)(n-f(n)-i)}/{n(n-1)}$.
    
    It follows that the expected number of interactions to aggregate all the data of $L^c$ is     
    \begin{align*}
        E(X) &= \sum_{i = 1}^{n - f(n)-1}\frac{n(n-1)}{2f(n)(n-f(n)-i)} \\
        &= \frac{n(n-1)}{2f(n)} \sum_{i = 1}^{n - f(n)-1}\frac{1}{n-f(n)-i}\\
        &\sim_{+\infty} \frac{n^2}{2f(n)}\log(n - f(n)) \\
        &=\frac{n^2}{2f(n)}\log(n(1 - f(n)/n)) \sim
        \frac{n^2\log(n)}{2f(n)}
    \end{align*}
    And the variance is
    \begin{align*}
        Var(X)&=\sum^{n - f(n)-1}_{i=1}\frac{ \left(1-\frac{2f(n)(n-f(n)-i)}{n(n-1)}\right)}{\left(\frac{2f(n)(n-f(n)-i)}{n(n-1)}\right)^2}\\
    %
    &\sim \sum^{n - f(n)-1}_{i=1} \frac{n^4}{4f(n)^2n^2}\sim \frac{n^3}{4f(n)^2}
    \end{align*}      
    Using the Chebyshev's inequality, we have
    \begin{align*}
        P\left(|X - E(X)| > \frac{n^2\log(n)}{2f(n)}\right)
        {=}O\left( \frac{1}{n\log^2(n)}\right)
    \end{align*} 
    Thus $X {=} O\left(\frac{n^2\log(n)}{f(n)}\right)$ with high probability.
%
\end{proof}

\begin{corollary}
    The algorithm Waiting Greedy, with $\tau = \Theta(n^{3/2}\sqrt{\log(n)})$  terminates in $\tau$ interactions with high probability.
\end{corollary}
\begin{proof}
    In the last theorem, the bound $O\left(\max\left(nf(n), n^2\log(n)/f(n)\right)\right)$ is minimized by the function $f: n \mapsto \sqrt{n\log(n)}$.
\end{proof}

\begin{theorem}\label{thm:optimality of greedy waiting}
    Waiting Greedy with $\tau = \Theta(n^{3/2}\sqrt{\log(n)})$ is optimal in $\DODA(meetTime)$.
\end{theorem}

\begin{proof}
    For the sake of contradiction, we suppose the existence of an algorithm $A\in \DODA(meetTime)$ that terminates in $T(n)$ interactions with high probability, with $T(n) = o\left(n^{3/2}\sqrt{\log(n)}\right)$. Without loss of generality we can suppose that $A$ does nothing after $T(n)$ interactions. Indeed, the algorithm $A'$ that executes $A$ up to $T(n)$ and does nothing afterward has the same upper bound (since the bound holds with high probability).

    Let $L$ be the set of nodes that interact directly with the sink during the first $T(n)$ interactions. Let $L^c$ be its complementary in $V\backslash\{s\}$. We know from Lemma~\ref{lemme: nf(n) interactions allow f(n) nodes to interact with the sink} that $\#L=O(T(n)/n)=o\left(\sqrt{n\log(n)}\right)$ \whp
    
    We can show that $T(n)$ interactions are not sufficient for all the nodes in $L^c$ to interacts with nodes in $L$. If nodes in $L^c$ want to send their data to the sink, some data must be aggregated among nodes in $L^c$, then the remaining nodes in $L^c$ that still own data must interact with a node in $L$ before $T(n)$ interactions (this is not even sufficient to perform the DODA, but is enough to reach a contradiction). 
    
    When two nodes in $L^c$ interact, their meetTime (that are greater than $T(n)$) and the previous interactions are independent with the future interactions occurring before $T(n)$. This implies that when two nodes in $L^c$ interact, using this information to decide which node transmits is the same as choosing the sender randomly. From corollary \ref{thm:optimality of gathering}, this implies that the optimal algorithm to aggregate data in $L^c$ is the Gathering algorithm. 
    
    Now, we show that, even after the nodes in $L^c$ use the Gathering algorithm, there is with high probability at least one node in $L^c$ that still owns data and that does not interact with any node in $L$. This node prevents the termination of the algorithm before $T(n)$ interactions with high probability, which is a contradiction.

    Formally, we have the following lemmas.
    \begin{lemma}
        Let $g(n)$ be the number of nodes in $L^c$. After using the Gathering algorithm during $T(n)$ interactions, the number of nodes in $L^c$ that still own data is in $\omega(\sqrt{n/\log(n)})$ with high probability.
    \end{lemma}
    \begin{proof}
        Let $X$ be the number of interactions needed for $R(n)$ nodes in $L^c$ to transmit their data.
        For the sake of contradiction, we suppose that 
        \begin{equation}\label{eq:g-R=o(g)}
            g(n) - R(n) = O(\sqrt{n/\log(n)}) = o(g(n))
        \end{equation}
        and show that $X$ is greater than $T(n)$ \whp 
        The probability of the $i$-th interaction between two nodes in $L^c$ that own data, after the $(i-1)$-th interaction already occurred, is $\frac{(g(n) - i)(g(n) - i - 1)}{n(n-1)}$. Thus we have:
        \begin{align*}
            E(X) &= \sum_{i=0}^{R(n) - 1}\frac{n(n-1)}{(g(n) - i)(g(n) - i - 1)}\\
                 &= n(n-1)\sum_{i=g(n) - R(n) + 1}^{g(n)}\frac{1}{i(i - 1)}\\
                 &= n(n-1)\left(\frac{1}{g(n) - R(n) + 1} - \frac{1}{g(n)}\right)\\
                 &= n(n-1)\frac{R(n)}{g(n)(g(n) - R(n))} 
        \end{align*}        
        From equation (\ref{eq:g-R=o(g)}) we deduce that $g\sim R$ and we have:
        \begin{align*}
            E(X)         &\sim n^2\frac{1}{g(n) - R(n)} 
        \end{align*}  
        which implies 
        \[
            E(X) = \Omega\left(n^2\sqrt{\frac{\log(n)}{n}}\right) = \Omega\left(n^{3/2}\sqrt{\log(n)}\right).
        \] As in the previous proofs, the expectation is reached with high probability. This contradicts the fact $T(n)=o(n^{3/2}\sqrt{\log(n)})$        
    \end{proof}
    \begin{lemma}
        Let $H\subset L^c$ be the nodes in $L^c$ that still own data after the gathering. With high probability, $T(n)$ interactions are not sufficient for all the nodes in $H$ to interact with nodes in $L$.
    \end{lemma}
    \begin{proof}
        We know from the previous lemma that the number of nodes in $H$ is $h(n) = \omega(\sqrt{n/\log(n)})$.
        Let $X$ be the random variable that equals the number of interactions needed for the nodes in $H$ to interact with the nodes in $L$. We show that $X$ is in $\omega(n^{3/2}\sqrt{\log(n)})$ with high probability. Indeed, the probability of the $i$-th interaction between a node in $H$ that owns data, after the $(i-1)$-th interaction already occurred, is $\frac{2f(n)(h(n) - i)}{n(n-1)}$, where $f(n)=\#L$. Thus we have:
        \begin{align*}
            E(X) &= \sum_{i=0}^{h(n) - 1}\frac{n(n-1)}{2f(n)(h(n) - i)}\\
                 &= \frac{n(n-1)}{2f(n)}\sum_{i=1}^{h(n)}\frac{1}{i}
                 \sim \frac{n^2}{2f(n)}\log(h(n))
        \end{align*}        
        But since $f(n) = o\left(\sqrt{n\log(n)}\right)$, we have
        \begin{align*}
            E(X) &= \omega\left(\frac{n^{3/2}}{\sqrt{\log(n)}}\log(h(n))\right) \\
            &= \omega\left( \frac{n^{3/2}}{\sqrt{\log(n)}}\log(n/log(n)) \right) \\
            &= \omega\left( n^{3/2}\sqrt{\log(n)} \right)
        \end{align*}
        Again the bound holds with high probability. This implies that, with high probability, $T(n)=o( n^{3/2}\sqrt{\log(n)} )$  interactions are not sufficient for all the nodes in $H$ to interact with nodes in $L$.
    \end{proof}
    \textit{End of the proof of theorem \ref{thm:optimality of greedy waiting}}.
    We have shown that $T(n)$ interactions are not sufficient for the nodes in $L^c$ to transmit their data (directly or indirectly) to the nodes in $L$. Indeed, we have shown that the nodes in $L^c$ can apply the gathering algorithm so that $\omega(\sqrt{n\log(n))}$ nodes in $L^c$ still own data with high probability. But, with high probability, one of the $\omega(\sqrt{n\log(n)})$ remaining nodes does not interact with a node in $L$ in $T(n)$ interactions. This implies that, with high probability, at least one node cannot send its data to the sink in $T(n)$ interactions and an algorithm $A$ with such a bound $T$ does not exist.    
\end{proof}


\section{Concluding remarks}

We defined and investigated the complexity of the distributed online data aggregation problem in dynamic graphs where interactions are controlled by an adversary. We obtained various tight complexity results for different adversaries and node knowledge, that open several scientific challenges:
\begin{enumerate}
\item What knowledge has a real impact on the lower bounds or algorithm efficiency ?
\item Can similar optimal algorithms be obtained with fixed memory or limited computational power ?
  \item Can randomized adversaries that use a non-uniform probabilistic distribution alter significantly the bounds presented here in the same way as in the work by Yamauchi \etal~\cite{YTKY12c}~?
  \end{enumerate}


\bibliographystyle{plain}
\bibliography{oda}

\end{document}